\documentclass[A4]{article}

\usepackage{amsmath,amsfonts,amsthm,amssymb,amscd,color}
\usepackage{slashbox, multirow}
\usepackage[dvipdfmx]{graphicx}
\usepackage[format=hang,font=footnotesize]{caption}

\theoremstyle{plain} \newtheorem{theorem}{Theorem}[section]
\newtheorem{lemma}[theorem]{Lemma}

 \theoremstyle{definition}
\newtheorem{definition}[theorem]{Definition} \theoremstyle{remark}
\newtheorem{remark}[theorem]{Remark} \newtheorem{example}[theorem]{Example}
 
\newcommand{\R}{{\mathbb R}}

\newcommand{\Z}{{\mathbb Z}}

\newcommand{\N}{{\mathbb N}}

 \def\im{{\rm i}}

\newcommand{\C}{\mathbb{C}} \newcommand{\T}{\mathbb{T}}

\newcommand{\stz}{\mathrm{Stz}}

\def\({\left(}
\def\){\right)}
\def\<{\left\langle}
\def\>{\right\rangle}

\newcommand{\Pa}{\partial}

\newcommand{\bra}[1]{ \left<  #1 \right>  }
\newcommand{\braa}[1]{\left(  #1 \right)  }
\newcommand{\brab}[1]{\left\{ #1 \right\} }

\newcommand{\abs}[1]{ \left|  #1 \right|  }
\newcommand{\nr}[1]{  \left\| #1 \right\| }

\newcommand{\wa}{W^\ast}
\newcommand{\wc}{\widetilde{C}_N}

\newcommand{\ve}[2]{
\begin{pmatrix}
#1
\\
#2
\end{pmatrix}
}

\numberwithin{equation}{section}

\begin{document}

\title{Scattering and inverse scattering for nonlinear quantum walks}

\author {Masaya Maeda, Hironobu Sasaki, Etsuo Segawa, Akito Suzuki, Kanako Suzuki}

\maketitle

\begin{abstract}
We study large time behavior of quantum walks (QWs) with self-dependent (nonlinear) coin.
In particular, we show scattering and derive the reproducing formula for inverse scattering in the weak nonlinear regime. 
The proof is based on space-time estimate of (linear) QWs such as dispersive estimates and Strichartz estimate.
Such argument is standard in the study of nonlinear Schr\"odinger equations and discrete nonlinear Schr\"odinger equations but it seems to be the first time to be applied to QW.
\end{abstract}

\section{Introduction}
Discrete time quantum walks (QWs) are space-time discrete unitary dynamics which can be considered to be the quantum analog of classical random walks \cite{ABNVW01,FH10Book,Gudder88Book,Meyer96JSP}.
QWs are attracting increasing interest for several reasons such as quantum search algorithms \cite{AKR05Proc,Childs09PRL,Portugal13Book}, 
model to understand topological insulators \cite{AO13PRB,CGSVWW16JPA,EKOS17JPA,GNVW12CMP,Kitaev06AP,Kitagawa12QIP,KRBD10PRA}, simulator of Dirac particles coupled with gauge and gravitational fields \cite{AD16PRA, AMBD16PRA, MBD14PA, MD12JMP,MDB15PRE, SFP15EPJQT}.
Moreover, QWs have been realized experimentally by optical lattice \cite{Karski09Science}, photons \cite{SCPGJS11,Schreiberetal12Science} and ion trappping \cite{ZKGSBR10PRL} (see \cite{MW14Book} for more reference).

In this paper, we consider the QWs with state dependent (nonlinear) quantum coin.
We set 
\begin{align*}
\mathcal H :=  l^2(\Z;\C^2):= \{ u:\Z\to \C^2\ |\ \|u\|_{l^2}^2:=\sum_{x\in \Z} \|u\|_{\C^2}^2<\infty\},
\end{align*}
where $\|u\|_{\C^2}=(|u_1|^2+|u_2|^2)^{1/2}$ for $u={}^t(u_1,u_2)\in \C^2$.
We fix a map $C:\R\times \R\to U(2)$, where $U(2)$ is the set of $2\times 2$ unitary matrices.
We define the (nonlinear) quantum coin $\hat C:\mathcal H\to \mathcal H$ by 
\begin{equation}
\label{defC}
(\hat Cu)(x)=C(|u_1(x)|^2,|u_2(x)|^2) u(x),
\end{equation} where $u={}^t(u_1\ u_2)\in \mathcal H$.
For $(T_\pm u)(x)=u(x\mp 1)$ and $S=\begin{pmatrix} T_- & 0 \\ 0 & T_+ \end{pmatrix},$
we set
\begin{align}\label{2}
U:=S\hat C:\mathcal H\to\mathcal H.
\end{align}
By definition, $S$ and $\hat C$ preserve the $l^2$ norm, 
and so does $U$.  
Let $u_0 \in \mathcal{H}$
be an initial state for a walker. 
Then, the state $u(t)$ of the walker at time $t$ is defined by the recursion 
relation 
\begin{align}\label{3}
u(t) = U u(t-1), \quad t\in \N,\quad u(0)=u_0.
\end{align}
We define the nonlinear evolution operator  $U(t)$ by
\[
U(t)u_0 =u(t), \quad t \in \N_0:=\N\cup\{0\}. 
\]
Notice that $u_0 \mapsto u = U(\cdot)u_0$ is a nonlinear map from 
$\mathcal H$ to $l_t^\infty (\N;\mathcal H)$.

By nonlinear QWs, we mean the nonlinear evolution generated by $U(t)$.
If $C:\R\times \R\to U(2)$ is a constant function (i.e.\ $C(s_1,s_2)=C_0\in U(2)$ for all $s_1,s_2\in \R$), then we will call it linear (or simple) QWs

\begin{remark}
One can generalize $C$ by defining it as a function $\Z\times \R\times \R$ and setting $(\hat C u)(x):=C(x,|u_1(x)|^2,|u_2(x)|^2)u(x)$.
In this paper we will only consider QWs which depend on its state but not explicitly on its position.
\end{remark}

To the best of authors knowledge, nonlinear QWs was first proposed by Navarrete, P\'erez and Rold\'an  \cite{NPR07PRA} as an nonlinear generalization of optical Galton board.

\begin{example}[Navarrete, P\'erez and Rold\'an  \cite{NPR07PRA} (see also \cite{MDB15PRE})] \begin{align}\label{4.2}
C(s_1,s_2)=\frac{1}{\sqrt{2}}\begin{pmatrix} 1 & 1\\ 1 & -1 \end{pmatrix} \begin{pmatrix} e^{\im g s_1} & 0 \\ 0 & e^{\im g s_2}\end{pmatrix},
\end{align}
where $g\in \R$.
\end{example}

As pointed out by Navarrete, P\'erez and Rold\'an  \cite{NPR07PRA} themselves,  
this nonlinear evolution does not define a quantum system,
but it can be realized in a optical system such as optical Galton board. 
Notice that this is similar to the relation between (linear) Schr\"odinger equation which describes quantum system and nonlinear Schr\"odinger equations which appears in various regions of physics including optics.
Moreover, in a way similar to linear QWs, 
we can define
the finding probability $p_t(x)$ of a walker at time $t$ at position $x$ 
through $u(t,x) := (U(t)u_0)(x)$ as
\begin{equation}
\label{find_prob}
p_t(x) = \|u(t,x)\|_{\mathbb{C}^2}^2, 
	\quad (t,x) \in \N_0 \times \mathbb{Z},
\end{equation}
provided that  $\|u_0\|_\mathcal{H} = 1$.
Indeed, $p_t$ gives a probability distribution on $\mathbb{Z}$,
because $U(t)$ preserves the norm and
$\sum_{x \in \mathbb{Z}} p_t(x) = \|u_0\|_{\mathcal{H}}^2$.
From these reasons, it is natural to view the system described by $U(t)$
as a nonlinearization of a linear QWs
and thus we simply call it 
a nonlinear QWs. 
We also import terminology from QWs and call  $\mathcal{H}$ and vectors in $\mathcal{H}$ the state space and states, respectively.
We prove a weak limit theorem for the nonlinear QWs 
in a companion paper \cite{MSSSSwlt}.

Other nonlinear QWs have been proposed by several authors as a simulator of nonlinear Dirac equation \cite{LKN15PRA}, for studying the nonlinear effect to the topologically protected mode \cite{GTB16PRA}, or simply investigating more rich dynamics (\cite{SWH14SR}).

\begin{example}[Lee, Kurzy\ifmmode~\acute{n}\else \'{n}\fi{}ski, and Nha \cite{LKN15PRA}]
The following models are proposed by the relation to nonlinear Dirac equations.
For Gross-Neveu model (scaler type interaction)
\begin{align}\label{4.4}
C(s_1,s_2) =\begin{pmatrix} e^{-\im g(s_1-s_2)} & 0 \\ 0 & e^{\im g(s_1-s_2)} \end{pmatrix}R(\theta),
\end{align}
and for Thirring model (vector type interaction)
\begin{align}\label{4.5}
C(s_1,s_2)=e^{\im g (s_1+s_2)}
R(\theta),
\end{align}
where $g,\theta\in\R$ and $R(\theta)=\begin{pmatrix} \cos\theta & -\sin\theta \\ \sin \theta & \cos \theta \end{pmatrix}$.
\end{example}

To construct a nonlinear QWs, it suffices to define $C:\R^2\to U(2)$.
Thus, the following nonlinear QWs is another natural example.
\begin{example}[Gerasimenko, Tarasinski and Beenakker \cite{GTB16PRA}]
Let $\theta:\R^2\to \R$.
We can define an nonlinear QWs by $C(s_1,s_2):=R(\theta(s_1,s_2))$.
In particular, setting $\theta(s_1,s_2):=\theta_0+g(s_1+\lambda s_2)^p$ with $\lambda=\pm1$, we obtain the nonlinear QWs with the nonlinear coin
\begin{align}\label{4.5.1}
C(s_1,s_2)=R(\theta_0+\lambda(gs_1+gs_2)^p)=R(\theta_0)R(\lambda(gs_1+gs_2)^p).
\end{align}
The particular nonlinear coin proposed in \cite{GTB16PRA} is the case $p=1$ and $\lambda=-1$.
\end{example}

In the following, we restrict our nonlinear coin operator to the following type:
\begin{align}
C(s_1,s_2) = C_0  C_N(gs_1,gs_2), \label{4.6}
\end{align}
where $g>0$ is a constant, 
\begin{align*}
C_0&=\begin{pmatrix} a & b \\-\bar b & \bar a \end{pmatrix}\in U(2),\  (|a|^2+|b|^2=1,\ 0<|a|<1),\\ C_N&\in C^2([0,\infty)\times [0,\infty);U(2)),
\end{align*}
and $C_N(0,0)=I_2$.
Here $I_2=\begin{pmatrix}  1 & 0 \\ 0 & 1 \end{pmatrix}$.
We set
\begin{align}\label{4.7}
U_0=S\hat C_0,\ \text{where}\ \(\hat C_0 u\)(x):=C_0u(x).
\end{align}

The positive parameter $g$ controls the strength of the nonlinearity.
Notice that all models \eqref{4.2}, \eqref{4.4}, \eqref{4.5} \eqref{4.5.1} given above are included in \eqref{4.6}.

In this paper, we view nonlinear QWs as space-time discretized nonlinear Schr\"odinger equations (NLS) and study the dynamical behavior of the walkers.
Indeed, we demonstrate that standard estimates such as dispersive estimate and Strichartz estimate hold for QWs (Theorem \ref{thm:1}, Lemma \ref{lem:stz}).
These estimates are fundamental tools for the study of NLS.
We show that also for nonlinear QWs, we can prove the scattering by parallel argument as the proof of scattering for NLS.
By scattering, we mean the following:

\begin{definition}\label{def:scatters}
We say $U(t)u_0$ scatters if there exists $u_+\in \mathcal H$ s.t.\ $\|U(t)u_0 - U_0^t u_+\|_{l^2}\to 0$ as $t\to \infty$, where $U_0=S C_0$.
\end{definition}

\begin{remark}
Scattering is equivalent to $$U(t)u_0 = U_0^t u_+  + \mathrm{Error}(t),\quad \|\mathrm{Error}(t)\|_{l^2}\to 0\ (t\to \infty).$$
Therefore, by scattering, we can conclude that the nonlinear QWs  behave similarly to linear QWs after long time.
However, $u_+$ will be generically different from $u_0$.
\end{remark}

We use $U_{g=1}(t)$ to denote the evolution $U(t)$
that has the nonlinear coin defined in \eqref{4.6} with $g=1$.
We observe that for $v_0 = \sqrt{g} u_0$ with $\|u_0\|_{l^2}=1$, 
\[ U(t)  u_0 = \frac{1}{\sqrt{g}} U_{g=1}(t) v_0. \]
%
Hence, instead of changing $g$, we can fix $g=1$ and vary the norm or $\|u_0\|_{l^2}$. 
Small $\|u_0\|_{l^2}$ will correspond to small $g$.
In the following, we will always fix $g=1$.

The first main result in this paper is the following:

\begin{theorem}\label{thm:scat}
Assume that $C_N\in C^1(\R^2;U(2))$ and there exists $c_0>0$ s.t.\ $\|C_N(s_1,s_2)-I_2\|_{\C^2\to \C^2}\leq c_0  (s_1+s_2)^{m}$ and  $\|\partial_{s_j}C_N(s_1,s_2)\|_{\C^2\to \C^2}\leq c_0  (s_1+s_2)^{m-1}$ for $j=1,2$.
Here, $\|A\|_{\C^2\to \C^2}$ is the operator norm of the matrix $A$.
That is $\|A\|_{\C^2\to \C^2}:=\sup_{v\in \C^2, \|v\|_{\C^2}=1}\|Av\|_{\C^2}$.
\begin{enumerate}
\item
For the case $m=3$, there exists $\delta>0$ s.t.\ for any $u_0\in l^2$ with $\|u_0\|_{l^2}<\delta$, $U(t)u_0$ scatters.
\item
For the case $m=2$, there exists $\delta>0$ s.t.\ for any $u_0\in l^1$ with $\|u_0\|_{l^1}<\delta$, $U(t)u_0$ scatters.
\end{enumerate}
\end{theorem}

Theorem \ref{thm:scat} tells us that if $u_0$ is sufficiently small (or with fixed $u_0$, $g$ is sufficiently small) the dynamics of nonlinear QWs will be similar to the dynamics of linear QWs.

We next consider the inverse scattering problem, 
which is the problem of identifying unknown nonlinear terms 
under the assumption that 
all of the scattering states are known.
More precisely, 
we identify some values concerning to $C_N$ 
using the scattering data $(u_0,W^\ast u_0)$.
Here,  
$W^\ast u_0$ is the final data $u_+$ appearing 
in Definition \ref{def:scatters}. 
Such problems naturally arise when one has only partial information of the system.
In this case, one would like to reconstruct the parameters governing the system from the data which one can observe.
In application, we usually do not have complete information of the system.
Therefore it is important to consider inverse scattering problems.
As for inverse scattering problems  
for some nonlinear Schr\"odinger equations and related equations,
there are many papers (see, e.g., 
\cite{%
CarlesGallagher2009,
MorawetzStrauss1973,
Sasaki2012,
Sasaki2015,
SasakiSuzuki2011,
Strauss1974,
Weder1997} 
and references therein).  
Using Theorem \ref{thm:scat} 
and modifying methods in the above papers, 
we obtain a reproducing formula for the nonlinear coin.

For simplicity, we consider the case that $C_N$ can be expressed as $C_N(s_1,s_2)=\tilde C_N(s_1^2,s_2^2)$ with 
\begin{align}\label{A}
\tilde C_N\in C^2(\R^2;U(2))\text{ and }\|\tilde C_N(s_1,s_2)-I_2\|_{\C^2\to \C^2}\leq C  |s_1|+|s_2|.
\end{align}

We define $\delta_{j,x}\in l^1(\Z;\C^2)$ by $\delta_{j,x}(y)= e_j$ if $y=x$ and $\delta_{j,x}(y)=0$ if $y\neq x$ where $e_1={}^t(1\ 0)$ and $e_2={}^t(0\ 1)$.
Further, for $g:\R_+\to \C$, we define $D_\lambda g(\lambda)=\lambda^{-1} \(g(2\lambda)-g(\lambda)\)$.
We define the nonlinear operator  
$W^\ast:u_0 \mapsto W^\ast u_0$.
It follows from the proof of Theorem \ref{thm:scat} that 
$W^\ast$ 
is well defined on $\{u_0\in l^1\ |\ \|u_0\|_{l^1}<\delta\}$  
and  satisfies 
\begin{align*}
W^* u_0 = u_0 + \sum_{t=0}^\infty U_0^{-t}\(\hat C_N-I_2\) U(t)u_0.
\end{align*}

\begin{theorem}[Inverse scattering]\label{thm:invscat}
Assume $(\mathrm{A})$ and that $\lambda>0$ is sufficiently small.
Then, we have
\begin{align*}
\left\| 
\begin{pmatrix} \mathcal L_{11}(\lambda)-\lambda D_\lambda \mathcal L_{11}(\lambda)&  D_\lambda \mathcal L_{11}(\lambda)\\
 \mathcal L_{12}(\lambda)-\lambda D_\lambda \mathcal L_{12}(\lambda)& D_\lambda \mathcal L_{12}(\lambda)
\end{pmatrix} -\partial_{1}\tilde C_N(0,0)\right\|_{\C^2\to \C^2}\leq C  \lambda^3,
\end{align*}
and
\begin{align*}
\left\| 
\begin{pmatrix} D_\lambda \mathcal L_{21}(\lambda)& 
\mathcal L_{21}(\lambda) -\lambda D_\lambda \mathcal L_{21}(\lambda)\\ D_\lambda \mathcal L_{22}(\lambda)&  \mathcal L_{22}(\lambda)-\lambda D_\lambda \mathcal L_{22}(\lambda)
\end{pmatrix} -\partial_{2}\tilde C_N(0,0)\right\|_{\C^2\to \C^2}\leq C  \lambda^3,
\end{align*}
where
\begin{align*}
\mathcal L_{1j}&=\lambda^{-10}\<\(U_0^{-1}W^* U_0-W^*\) \(\lambda^2 \delta_{1,0}+\lambda^3 \delta_{2,0}\), \delta_{j,0}\>,\\
\mathcal L_{2j}&=\lambda^{-10}\<\(U_0^{-1}W^* U_0-W^*\) \(\lambda^3 \delta_{1,0}+\lambda^2 \delta_{2,0}\), \delta_{j,0}\>.
\end{align*}
\end{theorem}

We note that Theorem \ref{thm:invscat} tells us that we can partially reconstruct the nonlinear coin from the information of the scattering states.
For example, the nonlinear coin is given by
\begin{align*}
\tilde C_N(s_1,s_2)=\begin{pmatrix} e^{\im g_1 s_1} & 0 \\ 0 & e^{\im g_2 s_2} \end{pmatrix},
\end{align*}
with some constants $g_1,g_2$, then we can recover
\begin{align*}
\tilde C_N(s_1,s_2)=I_2+\begin{pmatrix} \im g_1 s_1 & 0 \\ 0 & \im g_2 s_2 \end{pmatrix} + O(s_1^2+s_2^2),
\end{align*}
which will be the 1st order approximation of $\tilde C_N(s_1,s_2)$.
Therefore, we can identify the constants $g_1,g_2$ in the case.

The paper is organized as follows.
In section 2, we prove the dispersive estimate and Strichartz estimate for QWs with constant coin.
In section 3, we prove Theorem \ref{thm:scat}.
In section 4, we prove Theorem \ref{thm:invscat}.

\section{Dispersive and Strichartz estimates}

We first derive the dispersive estimate for the linear evolution $U_0=S\hat C_0$ by using stationary phase method.
We note that this dispersive estimate was first obtained by Sunada and Tate \cite{ST12JFA} in a slightly different form.

We define the (discrete) Fourier transform by
\begin{align}\label{9}
\mathcal F u (\xi):=\sum_{x\in \Z} e^{-\im x \xi}u(x),\quad \xi\in \T:=\R/2\pi \Z,
\end{align}
and 
 the inverse Fourier transform by
\begin{align}\label{10}
\mathcal F^{-1}f(x):=\frac{1}{2\pi}\int_\T e^{\im x \xi}f(\xi)\,d\xi.
\end{align}
Since $S\hat C_0u (x)= P_0 u(x+1)+Q_0u(x+1)$ where $P_0=\begin{pmatrix} a & b \\ 0 & 0 \end{pmatrix}$ and $Q_0 = \begin{pmatrix} 0 & 0 \\ -\bar b & \bar a \end{pmatrix}$, we have
\begin{align*}
\mathcal F(U_0u)(\xi)&=\sum_{x\in \Z}e^{-\im x \xi}\(P_0 u(x+1)+Q_0u(x-1)\)=\sum_{x\in \Z}e^{-\im x \xi}\(e^{\im \xi} P_0 +e^{-\im \xi}Q_0\)u(x)\\&=(e^{\im \xi} P_0 +e^{-\im \xi}Q_0)\mathcal F u(\xi).
\end{align*}
Notice that
\begin{align}\label{11}
\hat U_0(\xi) := e^{\im \xi}P_0+ e^{-\im \xi}Q_0=\begin{pmatrix}e^{\im \xi}a && e^{\im \xi}b\\ - \overline{e^{\im \xi}b} && \overline{e^{\im \xi}a} \end{pmatrix},
\end{align}
is also unitary and the
 eigenvalues are given by
\begin{align}\label{12}
\lambda_\pm(\xi)=w(\xi)\pm \im\sqrt{1-w(\xi)^2}=:e^{\pm\im \tilde p(\xi)},\quad w(\xi)=\mathrm{Re}(e^{\im \xi}a).
\end{align}
Thus diagonalizing $\hat U(\xi)$, we have
\begin{align}\label{12.0}
U_0=\mathcal F^{-1} P^{-1}(\xi)\mathrm{exp}\(\im \begin{pmatrix} \tilde p(\xi) & 0 \\ 0 & - \tilde p(\xi)\end{pmatrix}\)P(\xi)\mathcal F,
\end{align}
where
\begin{align}\label{12.01}
P(\xi)=\frac{1}{|b|^2+|e^{\im \xi}a-\lambda_+(\xi)|^2}
\begin{pmatrix}
-e^{-\im \xi}\bar b & -e^{-\im \xi}a + \lambda_-(\xi) \\ -e^{\im \xi} a + \lambda_+(\xi) & -e^{\im \xi} b
\end{pmatrix}.
\end{align}
We set $a=|a|e^{\im \theta_a}$.
Then, since we have $0<\tilde p<\pi$ by \eqref{12}, setting
\begin{align}\label{12.1}
p(\xi)=\mathrm{arccos}\(|a|\cos (\xi)\),
\end{align}
we have $\tilde p(\xi)=p(\xi+\theta_a)$.
Differentiating \eqref{12.1}, we obtain
\begin{align}\label{12.2}
p'(\xi) &= \frac{|a|\sin \xi}{\sqrt{1-|a|^2\cos^2\xi}},\\
p''(\xi)&=|a|(1-|a|^2)\frac{\cos\xi}{(1-|a|^2\cos^2\xi)^{3/2}},\label{12.3.2}\\
p'''(\xi)&=-|a|(1-|a|^2)\(1+2|a|^2\cos^2 \xi\)\frac{\sin \xi}{(1-|a|^2\cos^2\xi)^{5/2}}.\label{12.3.3}
\end{align}
By \eqref{12.0}, we have
\begin{align}
\(U_0^tu_0\)(x)&=
\(\(\frac{1}{2\pi}\int_{\T}P^{-1}(\xi)\exp\({\im t \begin{pmatrix} p(\xi+\theta_a)+\frac{\cdot}{t}\xi & 0 \\ 0 & -p(\xi+\theta_a)+\frac{\cdot}{t}\xi \end{pmatrix}}\) P(\xi)\,d\xi\)* u_0\)(x),\label{12.3.0}
\end{align}
where $A*u(x):=\sum_{y\in\Z} A(x-y)u(y)$.
We set the projections $P_\pm$ by
\begin{align*}
P_+:=\mathcal F^{-1}P^{-1}(\xi)\begin{pmatrix} 1 & 0 \\ 0 & 0 \end{pmatrix}P(\xi)\mathcal F,
\end{align*}
and $P_-=1-P_+$.
We define
\begin{align}\label{100}
I_\pm(t,s) = \frac{1}{2\pi}\int_\T e^{\im t\(\pm p(\xi)+s(\xi-\theta_a)\)}Q_\pm(\xi)\,d\xi,
\end{align}
where
\begin{align*}
Q_+=P^{-1}(\xi-\theta_a)\begin{pmatrix} 1 & 0 \\ 0 & 0 \end{pmatrix} P(\xi-\theta_a),\quad Q_-=P^{-1}(\xi-\theta_a)\begin{pmatrix} 0 & 0 \\ 0 & 1 \end{pmatrix} P(\xi-\theta_a).
\end{align*}
Then, we can express the generator by
\begin{align}\label{101}
U_0^t u_0 = \sum_{\pm}U_0^t P_\pm u_0=\sum_{\pm} I_\pm(t,\frac{\cdot}{t})*u_0.
\end{align}

The following is the dispersive estimate for QW.

\begin{theorem}\label{thm:1}
Let $0<|a|<1$.
Then, there exists $C>0$ such that for all $t\geq 1$,
\begin{align}\label{12.3}
\|U_0^t u \|_{l^\infty} \leq C t^{-1/3}\|u\|_{l^1}.
\end{align}
\end{theorem}

\begin{remark}
The constant $C$ given in Theorem \ref{thm:1} depends on $a$.
Moreover, the estimate \eqref{12.3} will not hold for $|a|=0,1$.
\end{remark}

\begin{proof}
The proof is similar to Theorem 3 of \cite{SK05N}.
By \eqref{101}, it suffices to show that for $t\geq 1$,
\begin{align}\label{12.3.1}
\sup_{s\in\R}\left\|  I_\pm(t,s) \right\|_{\C^2\to \C^2}\leq C\max_{1\leq i,j\leq 2}\sup_{s\in\R}|I_{\pm,ij}(t,s)|\leq C t^{-1/3},
\end{align}
where $I_\pm$ are given in \eqref{100} and $I_{\pm,ij}$ are the $(i,j)$ matrix component of $I_{\pm}$.
Here, we remark that $C$ in the middle of \eqref{12.3.1} and in the right hand side of \eqref{12.3.1} is different (so we are not claiming $\max_{1\leq i,j\leq 2}\sup_{s\in\R}|I_{\pm,ij}(t,s)|\leq  t^{-1/3}$).
We will use such conventions frequently.

From \eqref{12.3.2} and \eqref{12.3.3},
\begin{align*}
\(\frac{1+2|a|^2}{1-|a|^2}p''(\xi)\)^2+(p'''(\xi))^2&\geq
\(\frac{1+2|a|^2\cos^2\xi}{1-|a|^2\cos^2\xi}p''(\xi)\)^2+(p'''(\xi))^2\\&=|a|^2(1-|a|^2)^2\frac{\(1+2|a|^2\cos^2\xi\)^2}{(1-|a|^2\cos^2\xi)^{5}}\geq |a|^2(1-|a|^2)^2.
\end{align*}
This implies
$
\min_{\xi \in \T}(|p''(\xi)|, |p'''(\xi)|)>0.
$
Therefore, we can set $\psi_l\in C^\infty$ ($l=1,2$) s.t.\ $\psi_1(\xi)+\psi_2(\xi)=1$, $|p''(\xi)|\geq \delta$ for $\xi \in \mathrm{supp}\psi_1$ and $|p'''(\xi)|\geq \delta$ for $\xi \in \mathrm{supp}\psi_2$.
Now, Theorem \ref{thm:1} follows from Van der Corput lemma:

\begin{lemma}[Van der Corput lemma]\label{lem:1}
Let $\psi\in C^\infty$ and $k\geq 2$ and $|q^{(k)}(\xi)|\geq \delta$ in $\xi \in \mathrm{supp} \psi$.
Then, there exists $C>0$ (independent of $\delta>0$) s.t.\ we have
\begin{align}\label{12.3.5}
\left| \int_{\T}e^{\im t q(\xi)}\psi(\xi) \,d\xi\right|\leq C  (t \delta)^{-1/k},
\end{align}
for all $t>0$.
\end{lemma}

\begin{proof}
See \cite{SteinHarmonic}.
\end{proof}

\noindent
From Lemma \ref{lem:1}, we obtain the claim of Theorem \ref{thm:1}.
\end{proof}

As the case of Schr\"odinger equations and discrete Sch\"odinger equaiton (or continuous time QWs), we can derive the Strichartz estimate from dispersive estimate.
We define
\begin{align}\label{12.3.6}
\stz = l^\infty_t(\Z_{\geq 0};\l^2_x(\Z))\cap l^6_t(\Z_{\geq 0};l^\infty_x(\Z)),\quad \stz^*=l^1_t(\Z_{\geq 0};l^2_x(\Z))+l^{6/5}_t(\Z_{\geq 0};l^1_x(\Z)),
\end{align}
where $\|u\|_{l^p_t l^q_x}:=\(\sum_{t\geq 0} (\sum_{x\in \Z} |u(t,x)|^q)^{p/q}\)^{1/p}$ and 
\begin{align*}
\|u\|_{\stz}=\max(\|u\|_{l^\infty_t l^2_x}, \|u\|_{l^6_t l^\infty_x}),\quad \|u\|_{\stz^*}=\inf_{u_1+u_2=u}\(\|u_1\|_{l^1_t l^2_x}+\|u_2\|_{l^{6/5}_tl^1_x}\).
\end{align*}
We further define the weak $l^{p}$ space $l^{p,\infty}$ by its norm
\begin{align*}
\|f\|_{l^{p,\infty}}:=\sup_{\gamma>0} \gamma \(\#\{x\in \Z\ |\ |f(x)|>\gamma\}\)^{1/p},
\end{align*}
where $\#$ is the counting measure.
It is well known that $\|f\|_{l^{p,\infty}}\leq \|f\|_{l^p}$ and moreover we have $\|\<\cdot\>^{-1/p}\|_{l^{p,\infty}}<\infty$ ($\<x\>:=(1+|x|^2)^{1/2}$) and the Young's inequality for weak type spaces
\begin{align}\label{12.3.7}
\|f*g\|_{l^{p_0}}\leq C  \|f\|_{l^{p_1,\infty}}\|g\|_{l^{p_2}},
\end{align}
for $1<p_0,p_1,p_2<\infty$ and $1+p_0^{-1}=p_1^{-1}+p_2^{-1}$ (see Theorem 1.4.24 of \cite{GrafakosC}).

By parallel argument for the proof of Strichartz estimate of free Schr\"odinger equation, we have the following discrete Strichatrz estimate.

\begin{lemma}[Strichartz estimate]\label{lem:stz}
We have
\begin{align*}
\|U_0^t u_0\|_{\stz}\leq C\|u_0\|_{l^2},\quad \|\sum_{s=0}^{t} U_0^{t-s}f(s)\|_{\stz}\leq C\|f\|_{\stz^*},
\end{align*}
where $C>0$ is a constant.
\end{lemma}

\begin{proof}
We set $(p_\theta,q_\theta)$ by $(p_\theta^{-1},q_\theta^{-1})=(\frac{\theta}{6},\frac{1-\theta}{2})$ for $\theta\in [0,1]$.
By interpolation \cite{BLBook} between Theorem \ref{thm:1} and $l^2$ conservation, we have $\|U_0^t u_0\|_{l^{q_\theta}}\leq C\<t\>^{-\theta/3}\|u_0\|_{l^{q_\theta'}}$.
Here, $q_\theta' =\frac{q_\theta}{1-q_\theta}$ is the H\"older conjugate.
We first show the dual estimate:
\begin{align}\label{12.3.8}
\|\sum_{s=0}^\infty U_0^{-s}f(s)\|_{l^2}\leq C\|f\|_{\stz^*}.
\end{align}
Setting the inner product $\<\cdot,\cdot\>$ by $$\<f,g\>:=\sum_{x\in \Z}\<f(x),g(x)\>_{\C^2}=\sum_{x\in \Z}\(f_1(x)\overline{g_1(x)}+f_2(x)\overline{g_2(x)}\),$$
for $\theta\in[0,1]$, we have
\begin{align}
\|\sum_{s=0}^\infty U_0^{-s} f(s) \|_{l^2}^2 
&=\<\sum_{s=0}^\infty U_0^{-s}f(s),\sum_{t=0}^\infty U_0^{-t}f(t)\>
=\sum_{t=0}\<\sum_{s=0}^\infty U_0^{t-s}f(s),f(t)\>\label{12.3.9}
\\&\leq \sum_{t=0}^\infty \sum_{s=0}^\infty \|U_0^{t-s} f(s) \|_{l^{q_\theta}_x} \|f(t)\|_{l^{q_\theta'}_x}
\leq C \sum_{t=0}^\infty \(\sum_{s=0}^\infty \<t-s\>^{-\theta/3} \|f(s)\|_{l^{q_\theta'}_x}\) \|f(t)\|_{l^{q_\theta'}_x}\nonumber
\\&\leq C\left\|\<\cdot\>^{-\theta/3}* \|f(\cdot)\|_{l^{q_\theta'}_x} \right\|_{l^{p_\theta}_t} \|f\|_{l^{p_\theta'}_tl^{q_\theta'}_x}\leq C\|\<\cdot\>^{-\frac{\theta}{3}}\|_{l^{\frac{3}{\theta},\infty}}\|f\|_{l^{p_\theta'}_tl^{q_\theta'}_x}^2\leq C\|f\|_{l^{p_\theta'}_tl^{q_\theta'}_x}^2,\nonumber
\end{align}
where we have used \eqref{12.3.7} in the third line.
Therefore, we have \eqref{12.3.8}.
Notice that by the same argument we have
\begin{align}\label{12.3.10}
\|\sum_{s=0}^t U_0^{t-s}f(s)\|_{l^\infty_t l^2_x}+\|\sum_{s=t}^\infty U_0^{t-s}f(s)\|_{l^\infty_t l^2_x}\leq C \|f\|_{l^{p_\theta'}_tl^{q_\theta'}_x}.
\end{align}
The first claim follows from a duality argument using \eqref{12.3.8}.
\begin{align*}
\|U_0^t u_0 \|_{\stz}=\sup_{\|f\|_{\stz^*}\leq 1}  \sum_{t=0}^\infty \<U_0^t u_0, f\>=\sup_{\|f\|_{\stz^*}\leq 1}   \< u_0, \sum_{t=0}^\infty U_0^{-t} f\>\leq C\sup_{\|f\|_{\stz^*}\leq 1}\|u_0\|_{l^2} \|f\|_{\stz^*}\leq C\|u_0\|_{l^2}.
\end{align*}
We show the second inequality (inhomogeneous Strichartz).
For $\theta\in [0,1]$, applying the argument of \eqref{12.3.9}, we have
\begin{align}\label{12.3.11}
\|\sum_{s=0}^t U_0^{t-s} f(s)\|_{l^{p_\theta} l^{q_\theta}}\leq C \|\<\cdot\>^{-\frac{\theta}{3}}\|_{l^{\frac{3}{\theta},\infty}}\|f\|_{l^{p_\theta'}l^{q_\theta'}}\leq C \|f\|_{l^{p_\theta'}l^{q_\theta'}}.
\end{align}
Combining \eqref{12.3.10} and \eqref{12.3.11} applied for $\theta=1$, we have
\begin{align}\label{12.3.12}
\|\sum_{s=0}^t U_0^{t-s} f(s)\|_{\stz}\leq C  \|f\|_{l^{6/5}_tl^{1}_x}.
\end{align}
By \eqref{12.3.10}, for $(p,q)=(p_\theta,q_\theta)$,
\begin{align}
\|\sum_{s=0}^t U_0^{t-s}f(s)\|_{l^p_t l^q_x}&=\sup_{\|g\|_{l^{p'}_tl^{q'}_x}\leq 1} \sum_{t=0}^\infty\< \sum_{s=0}^t U_0^{t-s}f(s),g(t)\>=\sup_{\|g\|_{l^{p'}_tl^{q'}_x}\leq 1}\sum_{s=0}^\infty\<  f(s),\sum_{t=s}^\infty U_0^{s-t}g(t)\>\nonumber\\
&\leq \sup_{\|g\|_{l^{p'}_tl^{q'}_x}\leq 1}\|f\|_{l^1_tl^2_x} \|\sum_{t=s}^\infty U_0^{s-t}g(t) \|_{l^\infty_t l^2_x}\leq C \|f\|_{l^1_tl^2_x}.
\end{align}
Therefore, by interpolation, we have the conclusion.
\end{proof}

If we only use Strichartz estimate, we can only handle the case $\|C_N(s_1,s_2)\|_{\C^2\to \C^2}\leq C\(  s_1^3+s_2^3\)$.
To lower the power of the nonlinearity, we adapt the idea of Mielke and Patz \cite{MP10AA} (see also \cite{MP12DCDS}).

\begin{theorem}[Improved decay estimate]\label{thm:mp}
Let $0<|a|<1$.
Then, there exists $C>0$ s.t.\ we have
\begin{align}\label{12.4}
\| U_0^t u_0 \|_{l^{4,\infty}}\leq C \<t\>^{-1/4}\|u_0\|_{l^1}.
\end{align}
\end{theorem}

\begin{proof}
By \eqref{101} and Young's inequality for weak type spaces (see Theorem 1.2.13 of \cite{GrafakosC}):
\begin{align*}
\|f*g\|_{l^{p,\infty}}\leq C \|f\|_{l^{p,\infty}} \|g\|_{l^1},
\end{align*}
 it suffices to show that for $t\geq 1$,
\begin{align*}
\max_{1\leq i,j\leq 2}\left\| I_{\pm,ij}(t,\frac{\cdot}{t})\right\|_{l^{4,\infty}}\leq C t^{-1/4},\quad t\geq 1,
\end{align*}
where $I_{\pm}$ is given in \eqref{100} and $I_{\pm,ij}$ are the $(i,j)$ matrix component of $I_{\pm}$.
Further, for $|s|\geq 1$, we can show $|I(t,s)|\leq C (ts)^{-1}$ using integration by parts.
Therefore,
\begin{align*}
\|I_{\pm,ij}(t,\frac{\cdot}{t})\|_{l^4(|x|\geq t)}^4\leq C \sum_{|x|\geq t}|x|^{-4}\leq C t^{-3}.
\end{align*}
Thus, it suffices to show that for each $1\leq i,j\leq 2$, we have
\begin{align}\label{12.6}
\#\left\{x\in \Z_{|\cdot|\leq t} \ |\ \left| I_{\pm,ij}(t,\frac{x}{t})\right|>\gamma\right\}\leq C \gamma^{-4}t^{-1},
\end{align}
where $\Z_{|\cdot|\leq t}:=\{x\in \Z\ |\ |x|\leq t\}$.
Notice that if $\gamma\leq C  t^{-1/2}$, then \eqref{12.6} is automatically satisfied.
Thus, it suffices to consider the case $\gamma\gg t^{-1/2}$.
For $|s|\leq 1$, we claim
\begin{align}\label{13}
\left| I_{\pm,ij}(t,s)\right|\leq C  t^{-1/2}\(1+ \left|s^2 - |a|^2\right|^{-1/4}\).
\end{align}
\begin{remark}
This statement corresponds to Lemma 3.6 of \cite{MP10AA} and (4.12) of \cite{MP12DCDS}.
\end{remark}
If we have \eqref{13}, for $\gamma\gg t^{-1/2}$, we obtain \eqref{12.6} from
\begin{align*}
x\in  \left\{x\in \Z\ |\ \left| I_{\pm,ij}(t,\frac{x}{t})\right|>\gamma\right\} &\Rightarrow \gamma \leq C  t^{-1/2}\(1+ \left|(x/t)^2 - |a|^2\right|^{-1/4}\)  \Rightarrow \min_{\pm}(x \pm |a|t) \leq C  \gamma^{-4}t^{-1}.
\end{align*}
Thus, it suffices to show \eqref{13}.
Further, since we already have the global bound \eqref{12.3.1}, it suffices to consider the case
\begin{align}\label{14}
|s^2-|a|^2|\sim \min_{\pm}|s\pm |a| |\gtrsim t^{-2/3}.
\end{align}

For the case $|s|>|a|-t^{-3/2}$, we have $|\pm p'(\xi)+s|\geq |s|-|a|$.
Hence, integration by parts twice, we have
\begin{align*}
|I_{\pm,ij}(t,s)|\leq C  t^{-2}(|s|-|a|)^{-2}\leq C  t^{-1/2}\(1+ \left|s^2 - |a|^2\right|^{-1/4}\).
\end{align*}
For the case $|s|<|a|-t^{-3/2}$, we only consider $I_{+,11}$ and write $Q_{+,11}$ as $Q$ for simplicity.
Without loss of generality, we can assume $s\geq0$.
Set $\delta(s)>0$ so that $p'(-\frac{\pi}{2}\pm \delta(s))+s=0$.

We now fix $\delta_0\ll1$.
Then, if $\delta_0<|a|-|s|$, we have $|p'(\xi)+s|\geq \delta_1$ in $A_{\delta_1}:=\T\setminus \cup_{\pm}(-\frac{\pi}{2}\pm \delta(s)-\delta_1,-\frac{\pi}{2}\pm \delta(s)+\delta_1)$.
Thus, we have
\begin{align}
|I_{\pm,11}(t,s)|&\leq 
\left|\int_A e^{\im t(p(\xi)+s(\xi-\theta_a))}Q(\xi)\,d\xi\right|+\sum_{\pm}\left|\int_{-\frac{\pi}{2}\pm \delta(s)-\delta_1}^{-\frac{\pi}{2}\pm \delta(s)+\delta_1}e^{\im t(p(\xi)+s(\xi-\theta_a))}Q(\xi)\,d\xi\right|\nonumber\\&\leq C  \delta_1+(t \delta_1)^{-1}.\label{15}
\end{align}
Thus, if we take $\delta_1=t^{-1/2}$ in $A(\delta_1)$, we have \eqref{13}.

For the case $t^{-3/2}< |a|-|s|<\delta_0$, we claim 
\begin{align}\label{16}
|p'(\xi)+s|\gtrsim \(|a|^2-|s|^2\)^{1/2}\delta_1+\delta_1^2,\quad\text{for } \xi\in A(\delta_1).
\end{align}
If we have \eqref{16}, take $\delta_1=(|a|^2-|s|^2)^{-1/4}t^{-1/2}$.
Then, by \eqref{14}, we have $\delta_1\leq C  (|a|^2-s^2)^{1/2}$ and we can replace r.h.s.\ of \eqref{16} by $\(|a|^2-|s|^2\)^{1/2}\delta_1$.
Therefore, estimate of \eqref{15} with $(t \delta_1)^{-1}$ replaced by $(t(|a|^2-|s|^2)^{1/2}\delta_1)^{-1}$ will give us \eqref{13}.

Finally, we show \eqref{16}.
First, notice that as $s\to |a|$, $\delta(s)\to 0$.
Therefore, if $0<|a|-|s|<\delta_0\ll1$, we have $0<\delta(s)\ll1$ and
\begin{align*}
p'(-\pi/2 + \delta(s))+s= -|a|+s +\frac 1 2|a|(1-|a|^2)\delta(s)^2+o(\delta(s)^2),
\end{align*}
where we have used $p''(-\pi/2)=0$ and $p'''(-\pi/2)=|a|(1-|a|^2)$.
This gives us 
\begin{align}
\delta(s)= \sqrt{\frac{2(|a|-s)}{|a|(1-|a|^2)}}+o(\sqrt{|a|-s})
\end{align}
Since $\inf_{\xi\in A(\delta_1)}|p'(\xi)+s|$ is given by the minimum of $|p'\(-\pi/2 \pm(\delta(s)+\delta_1)\)+s|$, we have
\begin{align*}
|p'(-\pi/2 \pm(\delta(s)+\delta_1))+s|&=p'(-\pi/2 \pm \delta(s))+p''(-\pi/2 \pm \delta(s))(\pm \delta_1)+\frac 1 2 p'''(-\pi/2 \pm \delta(s))\delta_1^2\\&=
 \frac{1}{2}|a|(1-|a|^2)\(\delta(s)\delta_1+\delta_1^2\)+o(\delta_1^2). 
\end{align*}
Therefore, we have \eqref{16}.
\end{proof}

\section{Scattering}

We now prove Theorem \ref{thm:scat}

\begin{proof}[Proof of Theorem \ref{thm:scat} 1.]
We first estimate the Strichartz norm.
Set $\Phi:l^\infty_t l^2_x\to l^\infty_t l^2_x$ by
\begin{align*}
\Phi(u)(t)=U_0^t u_0+\sum_{s=0}^{t-1}U_0^{t-s} \(\hat C_N-I_2\)u(s).
\end{align*}
Notice that $U(t)u_0$ is the unique solution of \eqref{3} if and only if it is a fixed point of $\Phi$.
We show that if $\delta:=\|u_0\|_{l^2}\ll1$, then $\Phi$ has an fixed point.
Indeed, by lemma \ref{lem:stz}
\begin{align*}
\|\Phi(0)\|_{\stz}\leq C  \|u_0\|_{l^2},
\end{align*}
and
\begin{align*}
&\|\Phi(u)-\Phi(v)\|_{\stz}\leq C  \| U_0(\hat C-I_2)u-U_0(\hat C-I_2)v \|_{\stz^*}\leq C  \|(\hat C-I_2)u-(\hat C-I_2)v\|_{l^1l^2}\\&\leq C  \sum_{t\geq 0} \(\sum_{x\in \Z} \| C_N(|u_1(t,x)|^2,|u_2(t,x)|^2)-I_2\|_{\C^2\to\C^2}^2 \| u(t,x)-v(t,x)\|_{\C^2}^2\)^{1/2} \\&\quad+ C\sum_{t\geq 0} \(\sum_{x\in \Z} \|C_N(|u_1(t,x)|^2,|u_2(t,x)|^2)-C_N(|u_1(t,x)|^2,|u_2(t,x)|^2)\|_{\C^2\to\C^2}^2 \| v(t,x)\|_{\C^2}^2\)^{1/2}\\&\leq C 
\left\| \(\|u\|_{\C^2}^{6}+\|v\|_{\C^2}^{6}\) \| u-v\|_{\C^2}\right\|_{l^1_tl^2_x}\leq C  \left\| \|u\|_{\C^2}^6+\|v\|_{\C^2}^6 \right\|_{l^1_t l^\infty_x} \|u-v \|_{l^\infty_t l^2_x}
\\&
\leq C  
\(\|u\|_{\stz}^6+\|v\|_{\stz}^6\)\|u-v\|_{\stz}.
\end{align*}
Thus, if we set $\mathcal B:=\{ u\in \stz\ |\ \|u\|_{\stz}\leq 2 \|u_0\|_{l^2}=2 \delta\}$,
we see that $\Phi:\mathcal B\to \mathcal B$ is a contraction mapping, provided $\delta>0$ sufficiently small.
Therefore, there exists a unique $u$ s.t. $\Phi(u)=u$, which is actually $U(t)u_0$.
Further, since
\begin{align*}
U_0^{-t}U(t) u_0 = u_0 +\sum_{s=0}^{t-1}U_0^{-s}\(\hat C_N-I_2\)u(s),
\end{align*}
if we can show the right hand side is Cauchy in $\mathcal H$, we have the conclusion.
However, we already have the Strichartz bound $\|u\|_{\stz}\ll1$, so by the dual Strichartz estimate \eqref{12.3.8}, we have
\begin{align*}
\|\sum_{s=0}^{\infty} U_0^{-s}\(\hat C_N-I_2\)u(s)\|_{l^2}\leq C  \|u\|_{\stz}^7<\infty,
\end{align*}
and we conclude
\begin{align*}
\|\sum_{s=t_1}^{t_2}U_0^{-s}\(\hat C_N-I_2\)u(s)\|_{l^2}\to 0,\quad t_1\to \infty.
\end{align*}
Thus,
\begin{align*}
u_0 +\sum_{s=0}^{t-1}U_0^{-s-1}\(\hat C_N-I_2\)u(s)\to u_+\text{ in }l^2\text{ as }t\to \infty.
\end{align*}
Therefore, we have the conclusion.
\end{proof}

To show Theorem \ref{thm:scat} 2., we first show the decay of $l^5$ norm.
Notice that $\<t\>^{-4/15}$ is the same rate for the decreasing of linear solution, which is obtained by interpolation between $l^\infty$-$l^1$ estimate (Theorem \ref{thm:1}) and $l^{4,\infty}$-$l^1$ estimate (Theorem \ref{thm:mp}).
\begin{lemma}\label{lem:qdec}
Under the assumption of Theorem \ref{thm:scat} 2., there exists $\delta>0$ such that if $\|u(0)\|_{l^1}\leq \delta$, we have $\|U(t) u_0\|_{l^5} \leq C  \<t\>^{-4/15}\|u_0\|_{l^1}$.
\end{lemma}

\begin{proof}
We prove by induction.
Suppose $\|u(0)\|_{l^1}= \delta$.
Then, by Theorem \ref{thm:mp}, there exists $c_1>0$ s.t.\ $\|U_0^t u(0)\|_{l^5}\leq c_1 \<t\>^{-4/15} \|u(0)\|_{l^1}$.
We assume that for $0\leq s\leq t-1$, we have $\|u(s)\|_{l^5}\leq 2c_1 \delta\<s\>^{-4/15}$.
Then, we have
\begin{align*}
\<t\>^{4/15}\|u(t)\|_{l^5}&\leq c_1 \delta + \sum_{s=0}^{t-1}\<t\>^{4/15}c_1\<t-s-1\>^{-4/15}\|u(s)^5\|_{l^1}\\&\leq
c_1 \delta +  (2 c_1 \delta)^5 c_1\sum_{s=0}^{t-1}\<t\>^{4/15}\<t-s-1\>^{-4/15}\<s\>^{-4/3}.
\end{align*}
Notice that 
\begin{align*}
\sum_{s=0}^{t-1}\<t\>^{4/15}\<t-s-1\>^{-4/15}\<s\>^{-3/4}\leq c_2
\end{align*}
with some absolute constant $c_2>0$.
Indeed,
\begin{align*}
\sum_{s=0}^{t-1}\<t-s-1\>^{-4/15}\<s\>^{-3/4}&\sim \int_0^{t}\frac{1}{(1+t-s)^{4/15}}\frac{1}{(1+s)^{4/3}}
\\&\leq C  \<t\>^{-4/15}\int_0^{t/2}\<s\>^{-4/3}+\<t\>^{-4/3}\int_{t/2}^t (1+t-s)^{-4/15}\,ds
\\&\leq C  \<t\>^{-4/15}+\<t\>^{-4/3}\<t\>^{-4/15+1}\leq C   \<t\>^{-4/15}.
\end{align*}
Thus, if we take $0<\delta$ to satisfy $\delta<(2c_1 c_2^{1/5})^{-1}$, we have
\begin{align*}
c_1 \delta(1 +  32c_1^5\delta^5	c_2)< 2c_1 \delta.
\end{align*}
Therefore, we have the conclusion.
\end{proof}

We can show scattering by using decay.

\begin{proof}[Proof of Theorem \ref{thm:scat} 2.]
Let $\|u(0)\|_{l^1}\leq \delta$, where $\delta>0$ given by Lemma \ref{lem:qdec}. Then we have $\|u(t)\|_{l^5}\leq C  \<t\>^{-4/15}$.
Thus, $\|u\|_{l^{24/5}l^5(t,\infty)}\to 0$ as $t\to \infty$.
Therefore, we have
\begin{align*}
\|u\|_{\stz(T,\infty)}&\leq C  \|u(T)\|_{l^2}+C\|\sum_{s=T}^{t-1} U_0^{t-s}(\hat C_N-I_2)u(s)\|_{\stz}\\& \leq C  \|u(T)\|_{l^2} + C\| U_0(\hat C-I_2)u\|_{l^{6/5}l^1}\\&
\leq C  \|u(T)\|_{l^2} + C\|u\|_{l^{24/5}l^8(T,\infty)}^4 \|u\|_{l^\infty l^2}\\&
\leq C  \|u(T)\|_{l^2} + C\|u\|_{l^{24/5}l^5(T,\infty)}^4 \|u\|_{\stz}
\end{align*}
Therefore, we see that $\stz$ norm is finite.
Since we can bound $\|\sum_{s=0}^\infty U_0^{-s}(\hat C_N-I_2)u(s)\|_{l^2}$ by the same estimate, we have the conclusion.
\end{proof}

\section{Inverse scattering}
In this section,
we prove Theorem \ref{thm:invscat}.
Henceforth, 
we assume \eqref{A}.
we first prepare the following lemma:

\begin{lemma}\label{lem:2:thm:is:1}
Let $\delta>0$ sufficiently small.
Then, for any $u_0 \in  l ^1$ 
with $\|u_0\|_{ l ^1}\le \delta$,
\begin{align}\label{ineq:1:lem:2:thm:is:1}
\| U(t)u_0 - U_0^t u_0 \|_{ l ^\infty  l ^2}
\leq C  \|u_0\|_{ l ^1}^5
\end{align} 
and 
\begin{align}\label{ineq:2:lem:2:thm:is:1}
\| \(\hat C_N-I_2\) U(t)u_0 - \(\hat C_N-I_2\) U_0^t u_0 \|_{ l ^{6/5}  l ^1}
\leq C  \|u_0\|_{ l ^1}^9.
\end{align} 
\end{lemma}

\begin{proof}
By Lemma \ref{lem:stz}, 
Assumption \eqref{A} and the H\"older inequality, 
we have 
\begin{align*}
\| U(t)u_0 - U_0^t u_0 \|_{ l ^\infty  l ^2}
=&
\| \sum_{s=0}^{t-1}U_0^{t-s}\(\hat C_N-I_2\) U(s)u_0 \|_{ l ^\infty  l ^2}
\leq C 
\| \(\hat C_N-I_2\) U(t)u_0 \|_{ l ^{6/5} l ^1}
\\
\leq C 
&
\| | U(t)u_0 |^5 \|_{ l ^{6/5} l ^1}
\leq C 
\| U(t)u_0 \|_{ l ^{24/5} l ^5}^4 
\| U(t)u_0 \|_{ l ^\infty  l ^2}
\leq C 
\| u_0 \|_{ l ^1}^5.
\end{align*} 
Thus, we obtain \eqref{ineq:1:lem:2:thm:is:1}. 
\\

Let $u_1=\ve{u_{11}}{u_{12}} \in  l ^1$ 
and $u_2=\ve{u_{21}}{u_{22}} \in  l ^1$.
Then we have for any $x\in \Z$, 
\begin{align*}
&
\(\hat C_N-I_2\) u_1(x) - \(\hat C_N-I_2\) u_2(x)
\\
&=
\braa{\wc\braa{|u_{11}(x)|^4,|u_{12}(x)|^4}-I_2}u_1
-
\braa{\wc\braa{|u_{21}(x)|^4,|u_{22}(x)|^4}-I_2}u_2
\\
&=
\braa{\wc\braa{|u_{11}(x)|^4,|u_{12}(x)|^4}-I_2}(u_1(x)-u_2(x))
\\
&\quad +
\brab{
\braa{\wc\braa{|u_{11}(x)|^4,|u_{12}(x)|^4}-I_2}
-
\braa{\wc\braa{|u_{21}(x)|^4,|u_{22}(x)|^4}-I_2}
}u_2(x).
\end{align*} 
Using \eqref{A}, 
we obtain for any $x\in \Z$,
\begin{align*}
&
\abs{
\(\hat C_N-I_2\) u_1(x) - \(\hat C_N-I_2\) u_2(x)
}_{\C^2}
\\
&\leq C 
\braa{
|u_{11}(x)|^4+|u_{12}(x)|^4
}
\abs{ u_1(x)-u_2(x) }_{\C^2}
\\
&\quad + C
\braa{
|u_{11}(x)|^4-|u_{21}(x)|^4+|u_{21}(x)|^4-|u_{22}(x)|^4
}
|u_2(x)|_{\C^2}
\\
&\leq C 
\braa{
|u_1(x)|_{\C^2} + |u_2(x)|_{\C^2}
}^4
\abs{ u_1(x)-u_2(x) }_{\C^2}.
\end{align*} 
Hence it follows that 
\begin{align*}
&
\left\|
\(\hat C_N-I_2\) U(t)u_0 - \(\hat C_N-I_2\) U_0^t u_0 
\right\|_{ l ^{6/5} l ^1}
\\&\leq C 
\braa{
  \nr{U(t)  u_0}_{ l ^{24/5} l ^5} 
+ \nr{U_0^t u_0}_{ l ^{24/5} l ^5}
}^4
  \nr{ U(t)u_0-U_0^t u_0 }_{ l ^\infty  l ^2}
\leq C 
\nr{u_0}_{ l ^1}^9,
\end{align*} 
where we have used \eqref{ineq:1:lem:2:thm:is:1} 
and Lemma \ref{lem:qdec} in the last line.
This completes the proof of  \eqref{ineq:2:lem:2:thm:is:1}.
\end{proof}

\begin{proof}[Proof of Theorem \ref{thm:invscat}]
Let $u_0\in  l ^1$ with $\nr{u_0}_{ l ^1}\le \delta$ 
and let $v_0 \in  l ^1$.
We see from the proof of Theorem \ref{thm:scat} that 
the nonlinear operator $W^\ast$ satisfies 
\begin{align*}
\wa u_0 = u_0 + \sum_{t=0}^\infty U_0^{-t}\(\hat C_N-I_2\) U(t)u_0
\end{align*} 
and
\begin{align*}
\bra{ \wa u_0 - u_0 , v_0}
=
\sum_{t=0}^\infty \bra{ \(\hat C_N-I_2\) U(t)u_0 ,U_0^t v_0}.
\end{align*} 
By Lemmas \ref{lem:stz} and \ref{lem:2:thm:is:1}, 
we obtain 
\begin{align*}
&
\abs{
\bra{ \wa u_0 - u_0 , v_0}
-
\sum_{t=0}^\infty \bra{ \(\hat C_N-I_2\) U_0^t u_0 ,U_0^t v_0}
}
\\
&
\leq C 
\nr{ \(\hat C_N-I_2\) U(t)u_0-\(\hat C_N-I_2\) U_0^t u_0 }_{ l ^{6/5} l ^1}
\nr{ U_0^t v_0 }_{ l ^{6  } l ^\infty}
\\
&\leq C 
\nr{u_0}_{ l ^1}^9 
\nr{v_0}_{ l ^2}.
\end{align*}
Replacing $u_0$ and $v_0$ 
by $U_0u_0$ and $U_0v_0$, respectively, 
we have
\begin{align*}
&
\abs{
\bra{ U_0^{-1}\wa U_0 u_0 -u_0 , v_0}
-
\sum_{t=1}^\infty \bra{ \(\hat C_N-I_2\) U_0^t u_0 ,U_0^t v_0}
}
\\
&=
\abs{
\bra{ \wa U_0 u_0 - U_0 u_0 , U_0 v_0}
-
\sum_{t=0}^\infty \bra{ \(\hat C_N-I_2\) U_0^{t+1} u_0 ,U_0^{t+1} v_0}
}
\\
&\leq C 
\nr{U_0u_0}_{ l ^1}^9 
\nr{U_0v_0}_{ l ^2}
\leq C 
\nr{u_0}_{ l ^1}^9 
\nr{v_0}_{ l ^2},
\end{align*}
which implies that 
\begin{align*}
&
\abs{
\bra{ \braa{U_0^{-1}\wa U_0 -\wa}u_0 , v_0}
-
\bra{ \(\hat C_N-I_2\) u_0 ,v_0}
}
\leq C 
\nr{u_0}_{ l ^1}^9 
\nr{v_0}_{ l ^2}.
\end{align*} 
In particular, for $k=1,2$ and $l=1,2$
we see that 
\begin{align}
&
\left|
\< \braa{U_0^{-1}\wa U_0 -\wa}
 \(\lambda^{1+k}\delta_{1,0}+\lambda^{4-k}\delta_{2,0}\),
\delta_{l,0}\>
-
\< \(\hat C_N-I_2\) \(\lambda^{1+k}\delta_{1,0}+\lambda^{4-k}\delta_{2,0}\),
\delta_{l,0}\>
\right|
\nonumber\\
&\leq C 
\left\|\lambda^{1+k}\delta_{1,0}+\lambda^{4-k}\delta_{2,0}\right\|_{ l ^1}^9 
\|\delta_{l,0}\|_{ l ^2}
\leq C  
\lambda^{18}\label{ineq:1:proof}
\end{align} 
for any $\lambda>0$ sufficiently small.
By the Taylor theorem and \eqref{A}, 
we have 
\begin{align*}
&\< \(\hat C_N-I_2\) \(\lambda^{1+k}\delta_{1,0}+\lambda^{4-k}\delta_{2,0}\),
\delta_{l,0}\>
=
\bra{ \braa{\wc(\lambda^{4(1+k)},\lambda^{4(4-k)})-I_2}
\ve{\lambda^{1+k}}{\lambda^{4-k}},
\delta_{l,0}}_{\C^2}
\\
&=
\bra{ \Pa_1\wc(0,0)\lambda^{4(1+k)} \ve{\lambda^{1+k}}{\lambda^{4-k}},
\delta_{l,0}}_{\C^2}
+
\bra{ \Pa_2\wc(0,0)\lambda^{4(4-k)} \ve{\lambda^{1+k}}{\lambda^{4-k}},
\delta_{l,0}}_{\C^2} 
\\
&\quad +O(\lambda^{18})
\quad (\lambda\to+0).
\end{align*} 
We now define 
$\(\partial_k \tilde C_N(0,0)\)_{ij}=\<\partial_k \tilde C_N(0,0) e_j,e_i\>_{\C^2}$.
We see from (\ref{ineq:1:proof}) that 
\begin{align*}
\mathcal L_{1j}(\lambda)
&=
\lambda^{-10}
\< \(\hat C_N-I_2\) \(\lambda^{2}\delta_{1,0}+\lambda^{3}\delta_{2,0}\),
\delta_{j,0}\>
+O(\lambda^8)
\\
&=
\bra{ \Pa_1\wc(0,0)\ve{1}{\lambda},
\delta_{j,0}}_{\C^2}
+O(\lambda^4)
\\
&=
\(\Pa_1\wc(0,0)\)_{j1}+\lambda\(\Pa_1\wc(0,0)\)_{j2}
+O(\lambda^4)
\quad (\lambda\to+0)
\end{align*}
and
\begin{align*}
\mathcal L_{2j}(\lambda)
&=
\lambda^{-10}
\< \(\hat C_N-I_2\) \(\lambda^{3}\delta_{1,0}+\lambda^{2}\delta_{2,0}\),
\delta_{j,0}\>
+O(\lambda^8)
\\
&=
\bra{ \Pa_2\wc(0,0)\ve{\lambda}{1},
\delta_{j,0}}_{\C^2}
+O(\lambda^4)
\\
&=
\(\Pa_2\wc(0,0)\)_{j2}+\lambda\(\Pa_2\wc(0,0)\)_{j1}
+O(\lambda^4)
\quad (\lambda\to+0)
\end{align*}
for $j=1,2$.


Hence we have
\begin{align*}
\(\Pa_1\wc(0,0)\)_{j2}&=D_\lambda \mathcal L_{j1}+O(\lambda^3), \quad \(\Pa_2\wc(0,0)\)_{j1}=D_\lambda \mathcal L_{j2}+O(\lambda^3),\\
\(\Pa_1\wc(0,0)\)_{j1}&=\mathcal L_{1j}(\lambda)-\lambda D_\lambda \mathcal L_{1j}+O(\lambda^4), \quad \(\Pa_2\wc(0,0)\)_{j2}=\mathcal L_{2j}-\lambda D_\lambda \mathcal L_{2j}+O(\lambda^4),
\end{align*}
which completes the proof.
\end{proof}

\section*{Acknowledgments}  
M.M. was supported by the JSPS KAKENHI Grant Numbers JP15K17568, JP17H02851 and JP17H02853.
H.S. was supported by JSPS KAKENHI Grant Number JP17K05311.
E.S. acknowledges financial support from 
the Grant-in-Aid for Young Scientists (B) and of Scientific Research (B) Japan Society for the Promotion of Science (Grant No.~16K17637, No.~16K03939).
A. S. was supported by JSPS KAKENHI Grant Number JP26800054. 
K.S acknowledges JSPS the Grant-in-Aid for Scientific Research (C) 26400156.

\medskip

Masaya Maeda, Hironobu Sasaki

Department of Mathematics and Informatics,
Faculty of Science,
Chiba University,
Chiba 263-8522, Japan

{\it E-mail Address}: {\tt maeda@math.s.chiba-u.ac.jp, sasaki@math.s.chiba-u.ac.jp}

\medskip

Etsuo Segawa

Graduate School of Information Sciences, 
Tohoku University,
Sendai 980-8579, Japan

{\it E-mail Address}: {\tt e-segawa@m.tohoku.ac.jp}

\medskip

Akito Suzuki

Division of Mathematics and Physics,
Faculty of Engineering,
Shinshu University,
Nagano 380-8553, Japan

{\it E-mail Address}: {\tt akito@shinshu-u.ac.jp}

\medskip

Kanako Suzuki

College of Science, Ibaraki University,
2-1-1 Bunkyo, Mito 310-8512, Japan

{\it E-mail Address}: {\tt kanako.suzuki.sci2@vc.ibaraki.ac.jp}

\end{document}